\documentclass[10pt]{article}

\usepackage{hyperref}
\usepackage{times}
\usepackage{graphicx}
\usepackage{wrapfig}
\usepackage{color}
\usepackage{url}
\usepackage{cite}
\usepackage{amsmath}
\usepackage{amsthm,amssymb}
\usepackage{proof}
\usepackage[matrix,arrow,curve]{xy}

\allowdisplaybreaks

\theoremstyle{plain}
\newtheorem{theorem}{Theorem}
\newtheorem{lemma}[theorem]{Lemma}
\newtheorem{corollary}[theorem]{Corollary}

\theoremstyle{definition}
\newtheorem{definition}[theorem]{Definition}
\newtheorem{remark}[theorem]{Remark}
\newtheorem{notation}[theorem]{Notation}
\newtheorem{example}[theorem]{Example}

\newcommand{\entail}{\vdash}
\newcommand{\define}[1]{{\em #1}}

\newcommand{\vectoset}{{\rm VtoS}}

 \newcommand{\bor}{\ |\ }
\newcommand{\borsmall}{\,|\,}
\newcommand{\tensor}{\otimes}
\newcommand{\loli}{\mathbin{\multimap}}

\newcommand{\FinSet}{\mbox{\bf FinSet}}
\newcommand{\FinVec}{\mbox{\bf FinVec}}

\newcommand{\pair}[1]{{\langle\,{#1}\,\rangle}}
\newcommand{\punit}{{\star}}
\newcommand{\tunit}{{\bf 1}}
\newcommand{\bottype}{{\bf 0}}
\newcommand{\ttrue}{{\tt t}\!{\tt t}}
\newcommand{\ffalse}{{\tt f}\!{\tt f}}
\newcommand{\ifterm}[3]{{{\tt if}\,{#1}\,{\tt then}\,{#2}\,{\tt else}\,{#3}}}
\newcommand{\letunit}[2]{{\tt let}\,\star={#1}\,{\tt in}\,{#2}}
\newcommand{\bit}{{\tt Bool}}

\newcommand{\opeq}{\mathrel{\simeq_{\rm op}}}
\newcommand{\axeq}{\mathrel{\simeq_{\rm ax}}}

\newcommand{\defas}{\mathrel{{:}{=}}}

\newcommand{\fsdenot}[1]{{[\![{#1}]\!]^{\rm set}}}
\newcommand{\fvdenot}[1]{{[\![{#1}]\!]^{\rm vec}}}

\begin{document}

\title{Finite Vector Spaces as Model of Simply-Typed
  Lambda-Calculi\footnote{Accepted at ICTAC 2014. The final
    publication will be available at \url{http://link.springer.com}}}

\author{
  \scriptsize
  \begin{tabular}{c}
    {\large Beno\^it Valiron\footnote{Partially supported by the ANR project ANR-2010-BLAN-021301 LOGOI.}}\\[1.5ex]
    PPS, UMR 7126, Universit\'e Paris Diderot,\\
    Sorbonne Paris Cit\'{e}, F-75205 Paris, France
  \end{tabular}
  \quad
  \begin{tabular}{c}
    {\large Steve Zdancewic\footnote{This work
      has been funded in part by NFS grant CCF-1421193.}}\\[1.5ex]
    University of Pennsylvania,\\
    Philadelphia, US
  \end{tabular}
  \\[-1ex]
  ~}

\maketitle

\begin{abstract}
In this paper we use finite vector spaces (finite dimension, over
finite fields) as a non-standard computational model of linear
logic. We first define a simple, finite PCF-like lambda-calculus with
booleans, and then we discuss two finite models, one based on finite
sets and the other on finite vector spaces. The first model is shown
to be fully complete with respect to the operational semantics of the
language. The second model is not complete, but we develop an
algebraic extension of the finite lambda calculus that recovers
completeness. The relationship between the two semantics is described,
and several examples based on Church numerals are presented.
\end{abstract}

\section{Introduction}

A standard way to study properties of functional programming languages
is via denotational semantics. 
A denotational semantics (or model) for a language is a
mathematical representation of its programs~\cite{winskel-book}, and
the typical representation of a term is a function whose domain and
codomain are the
data-types of input and output.
This paper is concerned with a non-standard class of models
based on finite vector spaces.

The two languages we will consider are based on PCF~\cite{dana93pcf} -- the
laboratory mouse of functional programming languages. PCF comes as an
extension of simply-typed lambda-calculus with a call-by-name reduction
strategy, basic types and term constructs, and can be easily
extended to handle specific effects. Here, we define ${\bf PCF}_{\it f}$
as a simple lambda-calculus with pairs and booleans, and ${\bf PCF}_{\it f}^{\it alg}$, its
extension to linear combinations~of~terms.

There has been much work and progress on various denotational models
of PCF,
often with the emphasis on trying to achieve full abstraction. The
seminal works are using term models~\cite{milner77},
cpos~\cite{plotkin77} or game
semantics~\cite{abramski00}, while more recent works use
quantitative semantics of linear logic~\cite{ll}
and discuss probabilistic extensions~\cite{pcs}
or non-determinism~\cite{bucciarelli12}.

As a category, a model for a PCF language is at least required to be
cartesian closed to model internal morphisms and pairing. An
expressive class of cartesian closed categories can be made of models
of linear logic, by considering the (co)Kleisli category stemming from
the modality ``$!$''.
Although the models that are usually
considered are rich and expressive~\cite{finiteness,pcs,bucciarelli12},
``degenerate'' 
models nevertheless exist~\cite{pratt94chu,hyland03glueing}.
The consequences of the existence of such models of PCF
have not been explored thoroughly.

In this paper, we consider
two related finitary categories: the category of finite sets and
functions ${\FinSet}$ and the category of finite vector
spaces and linear functions ${\FinVec}$, i.e.
 finite-dimensional vector spaces over a finite
field. The adjunction between these two categories
 is known in the folklore to give a model of
linear logic~\cite{pratt92mail}, 
but the computational behavior of the corresponding coKleisli
category $\FinVec_!$ as a
model of PCF has not been studied until now.

The primary motivation for this work is simple curiosity: What do the
vectors interpreting lambda calculus terms look like?  Though not the
focus of this paper, one could imagine that the ability to encode
programming language constructs in the category of vector spaces might
yield interesting applications.  For instance, a Matlab-like
programming language that natively supports rich datatypes and
first-class functions, all with the same semantic status as
``vectors'' and ``matrices.''  A benefit of this design would be the
possibility of ``typed'' matrix programming, or perhaps sparse matrix
representations based on lambda terms and their semantics.  The
algebraic lambda calculus sketched in this paper is a
(rudimentary) first step in this direction.  Conversely, one could
imagine applying techniques from linear algebra to lambda calculus
terms.  For instance, finite fields play a crucial role in
cryptography, which, when combined with programming language semantics,
might lead to new algorithms for homomorphic encryption.

The goal here is more modest, however. The objective of the paper is to study how the two models
${\FinSet}$ and  $\FinVec_!$ fit
with respect to the language ${\bf PCF}_{\it f}$ and its algebraic extension
${\bf PCF}_{\it f}^{\it alg}$. In particular, we consider the usual three gradually 
more constraining properties: \textit{adequacy}, \textit{full abstraction} and
\textit{full completeness}.
A semantics is \textit{adequate} if whenever terms of some observable type
($\bit$ for example) are operationally equivalent then their
denotations match. An adequate semantics is ``reasonable'' in the
sense that programs and their representations match at ground type.
The semantics is \textit{fully abstract} if operational equivalence and
equality of denotation are the same thing for all types. In this
situation, programs and their denotations are in correspondence at all
types, but the model can contain non-representable
elements.
Finally, the semantics is \textit{fully complete} if moreover, every element in
the image of a type $A$ is representable by a term in the language.
With such a semantics, the set of terms and its mathematical
representation are fully correlated. If a semantics is fully complete,
then it is fully abstract and if it is fully abstract, then it is adequate.

\noindent
{\em Results.}~~
This paper presents the first account of the interpretation of two
PCF-like languages in finite vector spaces.
More specifically, we show that the category of finite sets
${\FinSet}$ forms a fully complete model for the language ${\bf PCF}_{\it f}$,
and that the coKleisli category ${\FinVec}_!$ is adequate but not
fully-abstract: this model has too many points compared to what one
can express in the language. We present several examples of the
encoding of Church numerals to illustrate the model.
We then present an algebraic extension ${\bf PCF}_{\it f}^{\it alg}$
of ${\bf PCF}_{\it f}$ and show that
${\FinVec}_!$ forms a fully complete model for this extension.
We discuss the relationship between the two languages and show how to
encode the extension within ${\bf PCF}_{\it f}$.

\noindent
{\em Related works.}~~
In the literature, finite models for lambda-calculi are commonly
used. For example, Hillebrand analyzes databases as finite
models of the simply-typed
lambda calculus~\cite{hillebrand-phd}. Salvati presents
a model based on finite sets~\cite{fin-model}, while
Selinger presents models based on finite
posets~\cite{fin-model-fin-posets}. Finally, Solovev~\cite{soloviev83}
relate the equational theory of cartesian closed categories
with the category of finite sets.

More general than vector spaces, various categories of modules over
semirings, as standard models of linear logic have been studied as
computational models: sets and relations~\cite{bucciarelli12}, finiteness
spaces~\cite{finiteness}, probabilistic coherent
spaces~\cite{pcs}, \textit{etc.}

As models of linear logic, 
finite vector spaces are folklore~\cite{pratt92mail} and
appear as side examples of more
general constructions such as Chu spaces~\cite{pratt94chu} or
glueing~\cite{hyland03glueing}.
Computationally, Chu spaces (and then to some extent finite vector
spaces) have been used in connection with automata~\cite{pratt94chu}.
Finally, recently finite vector spaces have also been used as a toy
model for quantum computation (see
e.g.~\cite{schumacher10modal,james11quantum}).

Algebraic lambda-calculi, that is, lambda-calculi with a vectorial
structure have been first defined in connection with finiteness
spaces~\cite{ehrhard03differential,vaux09}. Another
approach~\cite{ad08,adv11} comes to a similar type of language
from quantum computation. The former approach is call-by-name while
the latter is call-by-value. A general categorical
semantics has been developed~\cite{valiron13typed} but no other
concrete models have been considered.

\smallskip
\noindent
{\em Plan of the paper.}~~
The paper is shaped as follows. Section~\ref{sec:pcf} presents a
finite PCF-style language ${\bf PCF}_{\it f}$ with pairs and booleans, together
with its operational semantics. Section~\ref{sec:finset} presents the
category ${\FinSet}$ of finite sets and functions, and discusses its
properties as a model of the language
${\bf PCF}_{\it f}$. Section~\ref{sec:finvec} describes finite vector spaces and
shows how to build a model of linear logic from the adjunction with
finite sets. Section~\ref{sec:lc-finvec-model} discusses the
corresponding coKleisli category as a model of ${\bf PCF}_{\it f}$ and presents
some examples based on Church numerals. As ${\bf PCF}_{\it f}$ is not fully-abstract,
Section~\ref{sec:alglc} explains how to extend the language to better
match the model. Finally, Section~\ref{sec:discussion} discusses various 
related aspects: the
relationship between ${\bf PCF}_{\it f}$ and its extension, other categories in play, 
and potential generalization of fields.

\section{A finite PCF-style lambda-calculus}
\label{sec:pcf}

We pick a minimal finite PCF-style language with pairs and
booleans. We call it ${\bf PCF}_{\it f}$: it is intrinsically typed
(i.e. Church-style: all subterms are defined with their type) and
defined as follows.
\[
\begin{array}{@{}l@{~~~~}l@{~~~~}l@{}}
  M,N,P & {:}{:}{=} & 
  x \bor \lambda x.M \bor MN \bor
  \pi_l(M) \bor \pi_r(M) \bor \pair{M,N}\bor\punit\bor
  \\
  &&
  \ttrue \bor \ffalse \bor \ifterm{M}{N}{P}\bor\letunit{M}{N}
  \\
  A,B & {:}{:}{=} &
  \bit \bor A \to B \bor A \times B \bor \tunit.
\end{array}
\]
Values, including ``lazy'' pairs (that is, pairs of arbitrary terms,
as opposed to pairs of values), are inductively defined by 
$U,V$ ${:}{:}{=}$
$x \bor$ $\lambda x.M \bor \pair{M,N}\bor\punit\bor
  \ttrue \bor \ffalse$.
The terms consist of the regular lambda-terms, plus specific term
constructs. The terms $\ttrue$ and $\ffalse$ respectively stand for
the booleans True and False, while $\ifterm{-}{-}{-}$ is the boolean
test operator. The type $\bit$ is the type of the booleans. The term
$\star$ is the unique value of type $\tunit$, and $\letunit{-}{-}$
is the evaluation of a ``command'', that is, of a term evaluating to
$\star$. The term $\pair{-,-}$ is the pairing operation,
and $\pi_l$ and $\pi_r$ stand for the left and right projections. The
type operator $(\times)$ is used to type pairs, while $(\to)$ is used
to type lambda-abstractions and functions.

\begin{table}[t]
  \caption{Typing rules for the language ${\bf PCF}_{\it f}$.}
  \label{tab:typrule}
  \scalebox{0.79}{\begin{minipage}{6in}
  \[
  \infer{\Delta,x:A\entail x:A}{}
  \quad
  \infer{\Delta\entail\punit:\tunit}{}
  \quad
  \infer{\Delta\entail\ttrue,\ffalse:\bit}{}
  \quad
  \infer{\Delta\entail\lambda x.M:A\to B}{\Delta,x:A\entail M:B}
  \quad
  \infer{\Delta\entail\pi_i(M):A_i}{\Delta\entail M:A_l\times A_r}
  \]
  \[
  \infer{\Delta\entail MN:B}{\begin{array}{l}\Delta\entail M:A\to B \\
      \Delta\entail
    N:A\end{array}}
  ~~~
  \infer{\Delta\entail \pair{M,N}:A{\times}B}{\begin{array}{l}\Delta\entail M:A\\
    \Delta\entail N : B\end{array}}
  ~~~
  \infer{\Delta\entail \ifterm{M}{N_1}{N_2}:A}{\begin{array}{l}\Delta\entail M:\bit\\
    \Delta\entail N_1,N_2 : A\end{array}}
  ~~~
  \infer{\Delta\entail \letunit{M}{N}:A}{\begin{array}{l}\Delta\entail M:\tunit\\
    \Delta\entail N : A\end{array}}
  \]
  \end{minipage}}
\end{table}

A typing judgment is a sequent of the form $\Delta\entail M:A$, where
$\Delta$ is a typing context: a collection of typed variables $x:A$. A
typing judgment is said to be {\em valid} when there exists a valid typing
derivation built out of the rules in Table~\ref{tab:typrule}.

Note that since terms are intrinsically typed, for any valid typing
judgment there is only one typing derivation. Again because the terms
are intrinsically typed, by abuse of notation when the context is clear
we use $M:A$ instead of $\Delta\entail M:A$.

\begin{notation}
  When considering typing judgments such as $x:A\entail M:B$ and
  $y:B\entail N:C$, we use categorical notation to denote the
  composition: $M;N$ stands for the (typed) term $x:A\entail(\lambda
  y.N)M:C$, also
  written as $A\xrightarrow{M}B\xrightarrow{N}C$.
  We also extend pairs to finite
  products as follows: $\pair{M_1,M_2,\ldots}$ is the term
  $\pair{M_1,\pair{M_2,\pair{\ldots}}}$. Projections are generalized
  to finite products with the notation $\pi_i$ projecting the $i$-th
  coordinate of the product. Types are extended similarly:
  $A\times\cdots\times A$, also written as $A^{\times n}$, is defined
  as $A\times (A\times (\cdots))$. 
\end{notation}

\subsection{Small-step semantics}

\begin{table*}[b]
  \caption{Small-step semantics for the language ${\bf PCF}_{\it f}$.}
  \label{tab:cbn}
  \scalebox{.865}{\begin{minipage}{5.5in}
  \[
  \begin{array}{rcl}
    (\lambda x.M)N &\to& M[N/x]
    \\
    \letunit{\star}{M}&\to& M
  \end{array}
  \quad
  \begin{array}{rcl}
    \pi_l\pair{M,N}  &\to& M
    \\
    \pi_r\pair{M,N}  &\to& N
  \end{array}
  \quad
  \begin{array}{rcl}
    \ifterm{\ttrue}{M}{N} &\to& M
    \\
    \ifterm{\ffalse}{M}{N} &\to& N
  \end{array}
  \]
  \[
  \infer{MN\to M'N}{M\to M'}
  \qquad
  \infer{\pi_l(M)\to \pi_l(M')}{M\to M'}
  \qquad
  \infer{\pi_r(M)\to \pi_r(M')}{M\to M'}
  \]
  \[
  \infer{\ifterm{M}{N_1}{N_2}\to \ifterm{M'}{N_1}{N_2}}{M\to M'}
  ~~
  \infer{\letunit{M}{N}\to\letunit{M'}{N}}{M\to M'}
  \]
  \end{minipage}}
\end{table*}

The language is equipped with a call-by-name reduction strategy: a
term $M$ reduces to a term $M'$, denoted with $M\to M'$, when the
reduction can be derived from the rules of Table~\ref{tab:cbn}.
We use the notation $\to^*$ to refer to the reflexive transitive
closure of $\to$.

\begin{lemma}
  \label{lem:safety-prop}
  \label{lem:uniq-value}
  \label{lem:sn}
  (1) For any well-typed term $M:A$, either $M$ is a value or $M$ reduces
  to some term $N:A$. 
  ~~
  (2) The only closed value of type $\tunit$ is $\punit$ and the only
  closed values of type $\bit$ are $\ttrue$ and $\ffalse$.
  ~~
  (3) The language ${\bf PCF}_{\it f}$ is strongly normalizing.
\end{lemma}

\begin{proof}
  The fact that the language ${\bf PCF}_{\it f}$ is strongly
  normalizing comes from the fact that it can be easily encoded in the
  strongly normalizing language system F~\cite{sysF}.
\end{proof}

\subsection{Operational equivalence}
\label{sec:lc-cat}

We define the operational equivalence on terms in a standard way. A
{\em context $C[-]$} is a ``term with a hole'', that is, a
term consisting of the following grammar:
\[
\begin{array}{@{}l@{~}l@{~}l@{}}
  C[-] & {:}{:}{=} & 
  x \bor [-]\bor \lambda x.C[-] \bor C[-]N \bor MC[-] \bor\pi_l(C[-]) \bor
  \pi_r(C[-]) \bor
  \pair{C[-],N}\bor
  \\
  &&\pair{M,C[-]}\borsmall
  \punit\borsmall
  \ttrue \borsmall \ffalse \bor \ifterm{C[-]}{N}{P}\borsmall
  \ifterm{\!M\!}{C[-]}{P}\borsmall\\
  &&
  \ifterm{M}{N\!}{C[-]}\bor
  \letunit{C[-]}{M}\bor\letunit{M}{C[-]}.
\end{array}
\]
The hole can bind term variables, and a well-typed context is defined
as for terms. A closed context is a context with no free variables.

We say that $\Delta\entail M:A$ and $\Delta\entail N:A$ are
operationally equivalent, written $M\opeq N$, if for all closed
contexts $C[-]$ of type $\bit$ where the hole binds $\Delta$, for all
$b$ ranging over $\ttrue$ and $\ffalse$, $C[M]\to^* b$ if and only if
$ C[N]\to^* b $.

\subsection{Axiomatic equivalence}

We also define an equational theory for the language, called {\em
  axiomatic equivalence} and denoted with $\axeq$, and mainly
used as a technical apparatus.
The relation $\axeq$ is defined as the smallest reflexive, symmetric,
transitive and fully-congruent relation verifying the rules of
Table~\ref{tab:cbn}, together with the rule $\lambda
x.Mx \axeq M$ and the rule $\pair{\pi_l(M),\pi_r(M)} \axeq M$. 
A relation $\sim$ is said to be {\em fully-congruent}
on ${\bf PCF}_{\it f}$ if whenever $M\sim M'$, for all contexts $C[-]$
we also have $C[M]\sim C[N]$. The two additional rules are standard
equational rules for a lambda-calculus~\cite{lambek-scott}.

\begin{lemma}
  \label{lem:opax}
  If $M:A$ and $M\to N$ then $M\axeq N$.\qed
\end{lemma}

\section{Finite Sets as a concrete model}
\label{sec:lc-finset-model}
\label{sec:finset}

\begin{table*}[t]
  \caption{Denotational semantics for the language ${\bf PCF}_{\it f}$.}
  \label{tab:denot}
  \scalebox{.758}{\begin{minipage}{6.3in}
  \[
  \begin{array}{rcl}
    \fsdenot{\Delta,x:A\entail x:A} &:& (d,a) \longmapsto a
    \\[1.2ex]
    \fsdenot{\Delta\entail \ttrue:\bit} &:& d \longmapsto \ttrue
    \\[1.2ex]
    \fsdenot{\Delta\entail \ffalse:\bit} &:& d \longmapsto \ffalse
    \\[1.2ex]
    \fsdenot{\Delta\entail \punit:\tunit} &:& d \longmapsto \punit    
  \end{array}
  \quad
  \begin{array}{rcl}
    \fsdenot{\Delta\entail \pair{M,N}:A\times B} &:& d \longmapsto
    \pair{
      \fsdenot{M}(d),
      \fsdenot{N}(d)
    }
    \\[1.2ex]
    \fsdenot{\Delta\entail MN:B} &:& d \longmapsto
    \fsdenot{M}(d)(\fsdenot{N}(d))
    \\[1.2ex]
    \fsdenot{\Delta\entail\pi_l(M):A} &=&
    \fsdenot{M} ; \pi_l
    \\[1.2ex]
    \fsdenot{\Delta\entail\pi_r(M):B} &=&
    \fsdenot{M} ; \pi_r
  \end{array}
  \]
  \[
  \fsdenot{\Delta\entail\lambda x.M:A\to B} ~~=~~ d \longmapsto (a
  \mapsto \fsdenot{M}(d,a))
  \quad
  \fsdenot{\Delta\entail\letunit{M}{N}:A} ~~=~~ \fsdenot{N}
  \]
  \[
  \fsdenot{\Delta\entail\ifterm{M}{N}{P}:A} ~~=~~ d \longmapsto 
  \left\{
    \begin{array}{ll}
      \fsdenot{N}(d)
      &
      \textrm{if~~ $\fsdenot{M}(d) = \ttrue$,}
      \\
      \fsdenot{P}(d)
      &
      \textrm{if~~ $\fsdenot{M}(d) = \ffalse$.}
    \end{array}
  \right.
  \]
  \end{minipage}}
\end{table*}

Finite sets generate the full sub-category {\mbox{\FinSet}} of the category
{\bf Set}: objects are finite sets and morphisms are set-functions
between finite sets. The category is cartesian closed~\cite{soloviev83}: the product is
the set-product and the internal hom between two sets $X$ and $Y$ is
the set of all set-functions from $X$ to $Y$. Both sets are
finite: so is the hom-set.

We can use the category $\FinSet$ as a model for our
PCF language ${\bf PCF}_{\it f}$. The denotation of types corresponds to
the implicit meaning of the types: $\fsdenot{\tunit}\defas
\{\,\star\,\}$,
$\fsdenot{\bit} \defas \{\, \ttrue,\ffalse \,\}$, the product is the set-product
$\fsdenot{A\times B} \defas \fsdenot{A}\times\fsdenot{B}$, while the
arrow is the set of morphisms:
$\fsdenot{A\to B}\defas {\FinSet}(\fsdenot{A},\fsdenot{B})$.
The set $\{\ttrue,\ffalse\}$ is also written $\bit$. Similarly, the set
$\{\star\}$ is also written $\tunit$.
The denotation of a typing judgment $x_1:A_1,\ldots x_n:A_n\vdash
M:B$ is a morphism
$
\fsdenot{A_1}\times\cdots\times\fsdenot{A_n}\rightarrow
\fsdenot{B},
$
and is inductively defined as in Table~\ref{tab:denot}. 
The variable $d$ is assumed to
be an element of $\fsdenot{\Delta}$, while $a$ and $b$ are elements
of $\fsdenot{A}$ and $\fsdenot{B}$
respectively.

This denotation is sound with respect to the operational equivalence.

\begin{lemma}
  \label{lem:axdenot}
  If $M\axeq N:A$ then $\fsdenot{M} = \fsdenot{N}$.\qed
\end{lemma}

\begin{theorem}
  \label{th:sound}
  The model is sound with respect to the operational equivalence:
  Suppose that $\Delta\vdash M,N : A$.  If $\fsdenot{M} = \fsdenot{N}$
  then $M\opeq N$.
\end{theorem}

\begin{proof}
  Suppose that $M \not\opeq N$ and let $\Delta$ be $\{x_i:A_i\}_i$. 
  Then, because of Lemma~\ref{lem:sn},
  there exists a context $C[-]$
  such that $C[M]\to^*\ttrue$ and $C[N]\to^*\ffalse$. It follows that
  $(\lambda z.C[z\,x_1{\ldots} x_n])(\lambda x_1{\ldots} x_n.M)$
  $\axeq$
  $\ttrue$ and
  $(\lambda z.C[z\,x_1{\ldots} x_n])(\lambda x_1{\ldots} x_n.N)$
  $\axeq$
  $\ffalse$.
  If the denotations of $M$ and $N$ were equal, so would be the
  denotations of the terms $(\lambda x_1\ldots x_n.M)$ and $(\lambda x_1\ldots
  x_n.N)$. Lemmas~\ref{lem:opax} 
  and~\ref{lem:axdenot} yield a contradiction.
\end{proof}

${\FinSet}$ and the language ${\bf PCF}_{\it f}$ are somehow two
sides of the same coin. Theorems ~\ref{th:compl} and~\ref{th:eq} formalize
this correspondence.

\begin{theorem}[Full completeness]
  \label{th:compl}
  For every morphism $f:\fsdenot{A}\to\fsdenot{B}$ there exists a valid
  judgment $x:A\vdash M:B$ such that $f = \fsdenot{M}$.
\end{theorem}

\begin{proof}
  We start by defining inductively on $A$ two families of terms
  $M_a:A$ and $\delta_a:A\to\bit$ indexed by $a\in\fsdenot{A}$, such
  that $\fsdenot{M_a}=a$ and $\fsdenot{\delta_a}$ sends $a$ to
  $\ttrue$ and all other elements to $\ffalse$. For the types $\tunit$
  and $\bit$, the terms $M_\star$, $M_{\ttrue}$ and $M_{\ffalse}$ are
  the corresponding constants. The term $\delta_\star$ is $\lambda
  x.\star$, $\delta_{\ttrue}$ is $\lambda x.x$ while $\delta_{\ffalse}$ is
  the negation. For the type $A\times B$, one trivially calls the
  induction step. The type $A\to B$ is handled by remembering that the
  set $\fsdenot{A}$ is finite: if $g\in\fsdenot{A\to B}$, the term
  $M_g$ is the lambda-term with argument $x$ containing a list of {\tt
    if}-{\tt then}-{\tt else} testing with $\delta_a$ whether $x$ is
  equal to $a$, and returning $M_{g(a)}$ if it is. The term $\delta_g$
  is built similarly.
  The judgement $x:A\vdash M:B$ asked for in the theorem is obtained
  by setting $M$ to $(M_f)x$.
\end{proof}

\begin{theorem}[Equivalence]
  \label{th:eq}
  Suppose that $\Delta\vdash M,N : A$.  Then $\fsdenot{M} =
  \fsdenot{N}$ if and only if $M\opeq N$.
\end{theorem}

\begin{proof}
  The left-to-right implication is Theorem~\ref{th:sound}. We prove
  the right-to-left implication by contrapositive. Assume that
  $\fsdenot{M} \neq \fsdenot{N}$. Then there exists a function
  $f:\tunit\to\fsdenot{A}$ and a function
  $g:\fsdenot{B}\to\fsdenot{\bit}$ such that the boolean
  $f;\fsdenot{M};g$ is different from $f;\fsdenot{N};g$. By
  Theorem~\ref{th:compl}, the functions $f$ and $g$ are representable
  by two terms $N_f$ and $N_g$. They generate a context that
  distinguishes $M$ and $N$: this proves that $M \not\opeq N$.
\end{proof}

\begin{corollary}
  Since it is fully complete, the semantics $\FinSet$ is also adequate
  and fully abstract with respect to ${\bf PCF}_{\it f}$.\qed
\end{corollary}

\begin{example}\label{ex:numunit}
  Consider the Church numerals based over $\tunit$: they are
  of type
  $(\tunit\to\tunit)\to(\tunit\to\tunit)$. 
  In {\FinSet}, there is only one element since there is only one
  map from $\tunit$ to $\tunit$. 
  As a consequence of
  Theorem~\ref{th:eq}, one can conclude that all Church numerals
  $
  \lambda fx.f(f(\cdots (fx)\cdots))$ of type $(\tunit\to\tunit)\to(\tunit\to\tunit)$
  are operationally equivalent. Note that this is not true
  in general as soon as the type is inhabited by more elements.
\end{example}

\begin{example}\label{ex:numbool}
  How many operationally distinct Church numerals based over $\bit$
  are there\,?  From Theorem~\ref{th:eq}, it is enough to count how
  many distinct denotations of Church numerals there are in
  $\fsdenot{(\bit\to\bit)\to(\bit\to\bit)}$. There
  are exactly 4 distinct maps $\bit\to\bit$. Written as pairs $(x,y)$
  when $f(\ttrue)=x$ and $f(\ffalse)=y$, the maps ${\it tt}$, ${\it
    tf}$, ${\it ft}$ and  ${\it ff}$ are respectively
  $(\ttrue,\ttrue)$,
  $(\ttrue,\ffalse)$,
  $(\ffalse,\ttrue)$ and
  $(\ffalse,\ffalse)$.

  Then, if the Church numeral $\bar{n}$ is written as a tuple
  $(\bar{n}({\it tt}),\bar{n}({\it tf}),\bar{n}({\it ft}),\bar{n}({\it
    ff}))$, we have
    $\bar{0} = ({\it tf},{\it tf},{\it tf},{\it tf})$,
    $\bar{1} {}= ({\it tt},{\it tf},{\it ft},{\it ff})$,
    $\bar{2} {}= ({\it tt},{\it tf},{\it tf},{\it ff})$,
    $\bar{3} {}= ({\it tt},{\it tf},{\it ft},{\it ff})$,
  and one can show that for all $n\geq 1$, $\fsdenot{\bar{n}} =
  \fsdenot{\bar{n+2}}$. There are therefore only 3 operationally
  distinct Church numerals based on the type $\bit$: the number
  $\bar{0}$, then all even non-null numbers, and finally all odd
  numbers.
\end{example}

\section{Finite Vector Spaces}
\label{sec:finvec}

We now turn to the second finitary model that we want to use for
the language ${\bf PCF}_{\it f}$: finite vector spaces. We first start by
reminding the reader about this algebraic
structure.

\subsection{Background definitions}

A {\em field}~\cite{lidl1997finite} $K$ is a commutative ring such that the unit
$0$ of the addition is distinct from the unit $1$ of the
multiplication and such all non-zero elements of $K$ admit an inverse
with respect to the multiplication. A {\em finite field} is a field of
finite size. The {\em characteristic} $q$ of a field $K$ is the
minimum (non-zero) number such that $1 + \cdots + 1 = 0$ ($q$
instances of $1$). If there is none, we say that the
characteristic is $0$. For example, the field of real numbers has
characteristic $0$, while the field $\mathbb{F}_2$ consisting of $0$
and $1$ has characteristic $2$. The {\em order} of a finite field
is the order of its multiplicative group.

A {\em vector space}~\cite{lang} $V$ over a field $K$ is an algebraic structure
consisting of a set $|V|$, a binary addition $+$ and
a scalar multiplication $(\cdot) : K\times V\to V$, satisfying the
equations of Table~\ref{tab:rw-alg} (taken unordered).
The {\em dimension} of a vector
space is the size of the largest set of independent vectors.  A
particular vector space is the vector space {\em freely generated from
  a space $X$}, denoted with $\langle X\rangle$: it consists of all
the formal finite linear combinations $\sum_i\alpha_i\cdot x_i$, where
$x_i$ belongs to $X$ and $\alpha_i$ belongs to $K$. To define a linear
map $f$ on $\langle X\rangle$, it is enough to give its behavior on each
of the vector $x\in X$: the image of $\sum_i\alpha_i\cdot x_i$ is
then by linearity imposed to be $\sum_i\alpha_i\cdot f(x_i)$.

In this paper, the vector spaces we shall concentrate on are {\em
  finite vector spaces}, that is, vector spaces of finite dimensions
over a finite field. For example, the $2$-dimensional space
$\mathbb{F}_2\times\mathbb{F}_2$ consists of the four vectors
$\left(\begin{smallmatrix}
    0\\0
  \end{smallmatrix}\right),
  \left(\begin{smallmatrix}
    0\\1
  \end{smallmatrix}\right),
  \left(\begin{smallmatrix}
    1\\0
  \end{smallmatrix}\right),
  \left(\begin{smallmatrix}
    1\\1
  \end{smallmatrix}\right)
$
and is a finite vector space. It is also the vector space freely
generated from the $2$-elements set $\{\ttrue,\ffalse\}$:
each vectors respectively corresponds to
$
0,\ttrue, \ffalse$, and $\ttrue + \ffalse.
$

Once a given finite field $K$ has been fixed, 
the category {\FinVec} has for objects finite vector spaces over $K$
and for morphisms linear maps between
these spaces. The category is symmetric monoidal closed: the tensor
product is the algebraic tensor product, the unit of the tensor is
$I=K=\langle\star\rangle$ and the internal hom between
two spaces $U$ and $V$ is the vector space of all linear
functions $U\loli V$ between $U$ and $V$. The addition and the scalar
multiplication over functions are pointwise.

\subsection{A linear-non-linear model}
\label{sec:adjunction}

It is well-known~\cite{maclane} that the category of finite sets and functions and
the category of finite vector spaces and linear maps form an adjunction 
\begin{equation}
\label{eq:adj}
\xymatrix{
  \FinSet\ar@/^1ex/[r]^{F} &\FinVec \ar@/^1ex/[l]^{G}.
}
\end{equation}
The functor $F$ sends the set $X$ to the vector space $\langle X \rangle$ freely generated from
$X$ and the set-map $f:X\to Y$ to the linear map sending a basis
element $x\in X$ to the base element $f(x)$. The functor $G$ sends a
vector space $U$ to the same space seen as a set, and consider any
linear function as a set-map from the corresponding sets.

This adjunction makes {\FinVec}
into a model of linear logic~\cite{pratt92mail}. Indeed, the adjunction is symmetric
monoidal with the following two natural transformations:
\[
\begin{array}{l@{}lll}
m_{X,Y} :{}&\langle X\times Y\rangle &\to& \langle X\rangle\otimes \langle
Y\rangle
\\
&(x,y) &\mapsto& x\tensor y,
\end{array}
\quad
\begin{array}{l@{}lll}
m_{1} : {}&\langle 1\rangle &\to& I
\\
&\star&\mapsto& 1.
\end{array}
\]
This makes a {\em linear-non-linear category}~\cite{benton}, equivalent to a
linear category, and is a model of intuitionistic linear
logic~\cite{bierman}.

\subsection{Model of linear logic}

\begin{table*}[b]
  \caption{Modeling the language ${\bf PCF}_{\it f}$ in {\FinVec}.}
  \label{tab:denot-lc-alg}
  \scalebox{.75}{\begin{minipage}{6.3in}
  \[
  \begin{array}{rcl}
    \fvdenot{\Delta,x:A\entail x:A} &:& d\tensor b_a \longmapsto a
    \\[1.2ex]
    \fvdenot{\Delta\entail \ttrue:\bit} &:& d \longmapsto \ttrue
    \\[1.2ex]
    \fvdenot{\Delta\entail \ffalse:\bit} &:& d \longmapsto \ffalse
    \\[1.2ex]
    \fvdenot{\Delta\entail \punit:\tunit} &:& d \longmapsto \punit    
  \end{array}
  \quad
  \begin{array}{rcl}
    \fvdenot{\Delta\entail \pair{M,N}:A\times B} &:& d \longmapsto
    \fvdenot{M}(d)\tensor
    \fvdenot{N}(d)
    \\[1.2ex]
    \fvdenot{\Delta\entail MN:B} &:& d \longmapsto
    \fvdenot{M}(d)(\fvdenot{N}(d))
    \\[1.2ex]
    \fvdenot{\Delta\entail\pi_l(M):A} &=&
    \fvdenot{M} ; \pi_l
    \\[1.2ex]
    \fvdenot{\Delta\entail\pi_r(M):B} &=&
    \fvdenot{M} ; \pi_r
  \end{array}
  \]
  \[
  \begin{array}{rcl}
    \fvdenot{\Delta\entail\lambda x.M:A\to B} &=& d \longmapsto (b_a
    \mapsto \fvdenot{M}(d\tensor b_a))
    \\[1.5ex]
    \fvdenot{\Delta\entail\letunit{M}{N}:A} &=& d \longmapsto 
    \alpha\cdot\fvdenot{N}(d)
    \qquad\textrm{where~~ $\fvdenot{M}(d) =
      \alpha\cdot\star$.}
    \\[1.5ex]
    \fvdenot{\Delta\entail\ifterm{M}{N}{P}:A} &=& d \longmapsto 
    \alpha\cdot\fvdenot{N}(d) + 
    \beta\cdot\fsdenot{P}(d)
    \\[1.2ex]
    &&
    \hspace{0ex}
    \textrm{where~~ $\fvdenot{\Delta\entail M:\bit}(d) =
      \alpha\cdot\ttrue+\beta\cdot\ffalse$.}
  \end{array}
  \]\end{minipage}}
\end{table*}

The adjunction in Eq.~\eqref{eq:adj}
generates a linear comonad on ${\FinVec}$. If $A$ is a
finite vector space, we define the finite vector space $!A$ as the
vector space freely generated from the set $\{b_v\}_{v\in A}$:
it consists of the space
$\langle b_v \bor v\in A \rangle$.
If $f:A\to B$ is a linear map, the map ${!f}:{!A}\to{!B}$ is defined as
$
b_v \mapsto b_{f(v)}.
$
The comultiplication and the counit of the comonad are respectively 
$
\delta_A : !A \to {!!A}
$
and
$
\epsilon_A :  {!A}  \to  A
$
where
$\delta_A(b_v) = b_{b_v}$ and $\epsilon_A(b_v)=v$.
Every element $!A$ is a commutative comonoid when equipped with the
natural transformations
$\Delta_A : {!A}  \to {!A} \tensor {!A}$ and
$\Diamond_A  : {!A}  \to  I$
where
$\Delta_A(b_v) = b_v\tensor b_v$ and $\Diamond(b_v) =  1$.
This makes the category $\FinVec$ into a linear category.

In particular, the coKleisli category ${\FinVec}_{!}$ coming from the
comonad is cartesian closed: the product of $A$ and $B$ is
${A}\times{B}$, the usual product of vector spaces, and the terminal object is 
the vector space $\langle 0\rangle$.
This coKleisli category is the usual
one: the objects are the objects of ${\FinVec}$, and the morphisms
${\FinVec}_{!}(A,B)$ are the morphisms ${\FinVec}(!A,B)$.  The
identity $!A\to A$ is the counit and the composition of $f:{!A}\to B$
and $!B\to C$ is
$
f;g
\defas
{!A}\xrightarrow{\delta_A}{!!A}\xrightarrow{!f}{!B}\xrightarrow{g}C.
$

There is a canonical full embedding $E$ of categories sending
${\FinVec}_{!}$ on ${\FinSet}$.
It sends an object $U$ to the set of vectors of $U$ (i.e. it acts
as the forgetful functor on objects) and sends the linear map $f:{!U}\to V$ to the map
$v\mapsto f(b_v)$.

This functor preserves the cartesian closed structure: the terminal
object 
$\langle0\rangle$ of
$\FinVec_!$ is sent to the set containing only $0$, that is, the singleton-set
$\tunit$. The product space $U\times V$ is sent to the set of vectors
$\{\pair{u,v}\bor u\in{}U, v\in{}V\}$, which is exactly the set-product of $U$
and $V$. Finally, the function space $!U\to V$ is in exact correspondence with
the set of set-functions $U\to V$.

\begin{remark}\label{rem:ext-prod}
  The construction proposed as side example by Hyland and Schalk~\cite{hyland03glueing} 
  considers finite vector spaces with a field of characteristic~$2$. There,
  the modality is built using the exterior product algebra, and it turns out to
  be identical to the functor we use in the present paper. Note though,
  that their construction does not work with
  fields of other characteristics.
\end{remark}

\begin{remark}
  Quantitative models of linear logic such as finiteness
  spaces~\cite{finiteness} are also based on vector spaces;
  however, in these cases the procedure to build a comonad
  does not play well with the finite dimension the vector spaces
  considered in this paper: the definition of the comultiplication
  assumes that the space $!A$ is infinitely dimensional.
\end{remark}

\subsection{Finite vector spaces as a model}
\label{sec:lc-finvec-model}

Since ${\FinVec}_!$ is a cartesian closed category, one can model
terms of ${\bf PCF}_{\it f}$ as linear maps. Types are interpreted as
follows. The unit type is $\fvdenot{\tunit} \defas \{\,
\alpha\cdot\star\bor \alpha\in K \,\}$. The boolean type is
$\fvdenot{\bit} \defas \{\, {\textstyle\sum_i} \alpha_i\cdot\ttrue +
\beta_i\cdot\ffalse \bor\alpha_i,\beta_i\in K\,\}$. The product is the
usual product space: $\fvdenot{A\times B} \defas
\fvdenot{A}\,\times\,\fvdenot{B}$, whereas the arrow type is
$\fvdenot{A\to B} \defas {\FinVec}({!\fvdenot{A}},\fvdenot{B})$.
A typing judgment $x_1:A_1,\ldots,x_n:A_n\vdash M:B$ is represented
by a morphism of $\FinVec$ of type
\begin{equation}
  \label{eq:map}
  !\fvdenot{A_1}\otimes\cdots\otimes{!}\fvdenot{A_n}\longrightarrow
  \fvdenot{B},
\end{equation}
inductively defined as in Table~\ref{tab:denot-lc-alg}. The
variable $d$ stands for a base element $b_{u_1}\tensor\ldots\tensor
b_{u_n}$ of $\fvdenot{\Delta}$, and $b_a$ is a base element of
$\fvdenot{A}$. The functions $\pi_l$ and $\pi_r$ are the left and
right projections of the product.

  Note that because of the equivalence between ${!(A\times B)}$ and
  ${!A}\otimes {!B}$, the map in Eq.~\eqref{eq:map} is a morphism of
  ${\FinVec}_!$, as desired.

\refstepcounter{theorem}\label{ex:numunit-alg}%
\medskip
\noindent
{\bf Example~\ref{ex:numunit-alg}.}~
In {\FinSet}, there was only one Church numeral based on type
  $\tunit$. In ${\FinVec}_!$, there are more elements in the
  corresponding space
  $!({!\tunit}\loli\tunit)\loli({!\tunit}\loli\tunit)$ and we get more
  distinct Church numerals.

  Assume that the finite field under consideration is the 2-elements
  field $\mathbb{F}_2=\{0,1\}$. Then
  $
  \fvdenot{\tunit} = \tunit = \{0\cdot\star,1\cdot\star\} = \{0,\star\}.
  $
  The space $!\tunit$ is freely generated from the vectors of $\tunit$: it
  therefore consists of just the four vectors $\{0,b_0,b_\star,b_0+b_\star\}$.
  The space of morphisms $\fvdenot{\tunit\to\tunit}$ is the space
  ${!\tunit}\loli\tunit$. It is generated by two functions:
  $f_0$ sending $b_0$ to $\star$ and $b_\star$ to $0$, and $f_\star$
  sending $b_v$ to $v$. The space 
  therefore also contains 4 vectors: $0$, $f_0$, $f_\star$ and
  $f_0+f_\star$. Finally, the vector space ${!}({!\tunit}\loli\tunit)$
  is freely generated from the 4 base elements $b_0$, $b_{f_0}$,
  $b_{f_\star}$ and $b_{f_0+f_\star}$, therefore containing $16$
  vectors. Morphisms
  $!({!\tunit}\loli\tunit)\loli({!\tunit}\loli\tunit)$ can be
  represented by $2\times 4$ matrices with coefficients in
  $\mathbb{F}_2$. 
\begin{wrapfigure}{r}{0.24\textwidth}
\begin{minipage}{.2\textwidth}
  $\xy
  <1.5ex,0ex>:
  (1,0)*{\cdot};
  (0,1)*{\Big(\Big.};
  (1,2)*{\cdot};
  (2.5,0)*{\cdot};
  (2.5,2)*{\cdot};
  (4,0)*{\cdot};
  (4,2)*{\cdot};
  (5.5,0)*{\cdot};
  (5.5,2)*{\cdot};
  (6.5,1)*{\Big.\Big)};
  (8,0.1)*\txt{{\scriptsize $f_\star$}};
  (8,2.1)*\txt{{\scriptsize $f_0$}};
  (0.7,-1.2)*!L\txt{{\scriptsize $b_{0}$}};
  (2.2,-1.2)*!L\txt{{\scriptsize $b_{f_0}$}};
  (3.7,-1.2)*!L\txt{{\scriptsize $b_{f_\star}$}};
  (5.2,-1.2)*!L\txt{{\scriptsize $b_{f_0+f_\star}$}};
  \endxy$
\end{minipage}
\end{wrapfigure}
The basis elements $b_v$ are ordered as above, as
  are the basis elements $f_w$, as shown on the right.
  The Church numeral $\bar{0}$ sends all of its arguments to the
  identity function, that is, $f_\star$. The Church numeral $\bar{1}$
  is the identity. So their respective matrices are
  $\left(\begin{smallmatrix}
    0&0&0&0\\
    1&1&1&1
  \end{smallmatrix}\right)$
  and
  $\left(\begin{smallmatrix}
    0&1&0&1\\
    0&0&1&1
  \end{smallmatrix}\right)$.
  The next two Church numerals are
  $\bar{2} = 
  \left(\begin{smallmatrix}
    0&0&0&1
    \\
    0&1&1&1
  \end{smallmatrix}\right)$ and
  $\bar{3} = 
  \left(\begin{smallmatrix}
    0&1&0&1\\
    0&0&1&1
  \end{smallmatrix}\right)$, which is also $\bar{1}$.
  So ${\FinVec}_!$ with the field of characteristic $2$ distinguishes
  null, even and odds numerals over the type $!\tunit$.

  Note that this characterization is similar to the ${\FinSet}$ 
  Example~\ref{ex:numbool}, except that there, the type over which the
  Church numerals were built was $\bit$. Over $\tunit$,
  Example~\ref{ex:numunit} stated that all Church numerals collapse.

\begin{example}
  \label{ex:numunit-alg2}
  The fact that ${\FinVec}_!$ with the field of characteristic 2 can
  be put in parallel with {\FinSet} when considering Church numerals
  is an artifact of the fact that the field has only two elements. If
  instead one chooses another field
  $K=\mathbb{F}_p=\{0,1,\ldots,p-1\}$ of
  characteristic $p$, with $p$ prime,
  then this is in general not true anymore. In this case,
  $\fvdenot{\tunit}=\{0,\star,2\cdot\star,\ldots,(p-1)\cdot\star\}$, and ${!\tunit}\loli\tunit$
  has dimension $p$ with basis elements 
  $f_i$ sending $b_{i\cdot\star}\mapsto \star$ and $b_{j\cdot\star} \mapsto 0$ when $i\neq j$.
  It therefore consists of $p^p$ vectors. Let us represent a function
  $f:{!\tunit}\loli\tunit$ with ${\tt x_0\ldots{}x_{p-1}}$ where
  $f(b_{i\cdot\star})={\tt x}_i\cdot\star$. A morphism
  $!({!\tunit}\loli\tunit)\loli({!\tunit}\loli\tunit)$ can be
  represented with a $p^p\times p$ matrix. The basis elements $b_{\tt
    x_0\ldots{}x_{p-1}}$ of $!({!\tunit}\loli\tunit)$ are ordered
  lexicographically:
  $
  b_{0\ldots{}00}, b_{0\ldots01}, b_{0\ldots02}, \ldots, 
  b_{0\ldots0(p-1)}, \ldots, b_{(p-1)\ldots(p-1)},
  $
  as are the basis elements $f_0,f_1,\ldots,f_{p-1}$.
\begin{table*}[t]
  \caption{The $8$ Church numerals over type $\tunit$ in ${\FinVec}_!$ with $K=\mathbb{F}_3$.}
  \label{tab:sevenmat}
  \scalebox{.84}{\begin{minipage}{5.67in}
  \[
  \begin{array}{ll}
  \scalebox{.8}{$\overline{0} = {}$}\\ \left(
    \begin{smallmatrix}
      0&0&0&0&0&0&0&0&0&0&0&0&0&0&0&0&0&0&0&0&0&0&0&0&0&0&0\\
      1&1&1&1&1&1&1&1&1&1&1&1&1&1&1&1&1&1&1&1&1&1&1&1&1&1&1\\
      2&2&2&2&2&2&2&2&2&2&2&2&2&2&2&2&2&2&2&2&2&2&2&2&2&2&2
    \end{smallmatrix}
  \right)
  \end{array}
  \begin{array}{ll}
  \scalebox{.8}{$\overline{1} = {}$}\\ \left(
    \begin{smallmatrix}
      0&0&0&0&0&0&0&0&0&1&1&1&1&1&1&1&1&1&2&2&2&2&2&2&2&2&2\\
      0&0&0&1&1&1&2&2&2&0&0&0&1&1&1&2&2&2&0&0&0&1&1&2&2&2&2\\
      0&1&2&0&1&2&0&1&2&0&1&2&0&1&2&0&1&2&0&1&2&0&1&2&0&1&2
    \end{smallmatrix}
  \right)
  \end{array}
  \]
  \[
  \begin{array}{ll}
  \scalebox{.8}{$\overline{2+6n} = {}$}\\ \left(
    \begin{smallmatrix}
      0&0&0&0&0&0&0&0&0&0&0&0&1&1&1&2&2&2&0&1&2&0&1&2&0&1&2\\
      0&0&0&1&1&1&0&1&2&1&1&1&1&1&1&0&1&2&2&2&2&1&1&1&0&1&2\\
      0&0&2&0&1&2&0&2&2&1&0&2&1&1&2&1&2&2&2&0&2&2&1&2&2&2&2
    \end{smallmatrix}
  \right)
  \end{array}
  \begin{array}{ll}
    \scalebox{.8}{$\overline{3+6n} = {}$}\\ \left(
    \begin{smallmatrix}
      0&0&0&0&0&0&0&0&0&1&1&1&1&1&1&0&1&2&2&0&2&2&1&2&2&2&2\\
      0&0&0&1&1&1&0&2&2&0&0&0&1&1&1&1&2&2&0&1&2&1&1&1&2&2&2\\
      0&0&2&0&1&2&0&1&2&0&1&2&1&1&2&2&1&2&0&2&2&0&1&2&0&1&2
    \end{smallmatrix}
  \right)
  \end{array}
  \]
  \[
  \begin{array}{ll}
  \scalebox{.8}{$\overline{4+6n} = {}$}\\ \left(
    \begin{smallmatrix}
      0&0&0&0&0&0&0&0&0&0&0&0&1&1&1&1&2&2&0&2&2&0&1&2&0&1&2\\
      0&0&0&1&1&1&0&1&2&1&1&1&1&1&1&2&1&2&2&0&2&1&1&1&0&1&2\\
      0&0&2&0&1&2&0&2&2&1&0&2&1&1&2&0&2&2&2&1&2&2&1&2&2&2&2
    \end{smallmatrix}
  \right)
  \end{array}
  \begin{array}{ll}
    \scalebox{.8}{$\overline{5+6n} = {}$}\\ \left(
    \begin{smallmatrix}
      0&0&0&0&0&0&0&0&0&1&1&1&1&1&1&2&1&2&2&1&2&2&1&2&2&2&2\\
      0&0&0&1&1&1&0&2&2&0&0&0&1&1&1&0&2&2&0&2&2&1&1&1&2&2&2\\
      0&0&2&0&1&2&0&1&2&0&1&2&1&1&2&1&1&2&0&0&2&0&1&2&0&1&2
    \end{smallmatrix}
  \right)
  \end{array}
  \]
  \[
  \begin{array}{ll}
  \scalebox{.8}{$\overline{6+6n} = {}$}\\ \left(
    \begin{smallmatrix}
      0&0&0&0&0&0&0&0&0&0&0&0&1&1&1&0&2&2&0&0&2&0&1&2&0&1&2\\
      0&0&0&1&1&1&0&1&2&1&1&1&1&1&1&1&1&2&2&1&2&1&1&1&0&1&2\\
      0&0&2&0&1&2&0&2&2&1&0&2&1&1&2&2&2&2&2&2&2&2&1&2&2&2&2
    \end{smallmatrix}
  \right)
  \end{array}
  \begin{array}{ll}
    \scalebox{.8}{$\overline{7+6n} = {}$}\\ \left(
    \begin{smallmatrix}
      0&0&0&0&0&0&0&0&0&1&1&1&1&1&1&1&1&2&2&2&2&2&1&2&2&2&2\\
      0&0&0&1&1&1&0&2&2&0&0&0&1&1&1&2&2&2&0&0&2&1&1&1&2&2&2\\
      0&0&2&0&1&2&0&1&2&0&1&2&1&1&2&0&1&2&0&1&2&0&1&2&0&1&2
    \end{smallmatrix}
  \right)
  \end{array}
  \]\end{minipage}}
\end{table*}%

  The Church numeral $\bar{0}$ is again the constant function
  returning the identity, that is, $\sum_i i\cdot{}f_i$.  The numeral
  $\bar{1}$ sends ${\tt x_0\cdots{}x_{p-1}}$ onto the function sending
  $b_{i\cdot\star}$ onto ${\tt x}_i\cdot\star$. The numeral $\bar{2}$
  sends ${\tt x_0\cdots{}x_{p-1}}$ onto the function sending $b_{i\cdot\star}$
  onto ${\tt x}_{{\tt x}_i}\cdot\star$. The numeral $\bar{3}$ sends
  ${\tt x_0\cdots{}x_{p-1}}$ onto the function sending $b_{i\cdot\star}$ onto
  ${\tt x}_{{\tt x}_{{\tt x}_{{\tt x}_i}}}\cdot\star$. And so on.
  
  In particular, each combination ${\tt x_0}\cdots{}{\tt x}_{p-1}$ can
  be considered as a function ${\tt x}:\{0,\ldots p-1\}\to\{0,\ldots
  p-1\}$. The sequence $({\tt x}^0,{\tt x}^1,{\tt x}^2,\ldots)$
  eventually loops. The order
  of the loop is ${\it lcm}(p)$, the least common multiple of all
  integers $1,\ldots,p$, and for all $n\geq p-1$ we have ${\tt x}^n={\tt
    x}^{n+{\it lcm}(p)}$: there are ${\it lcm}(p) + p - 1$ distinct
  Church numerals in the model
  ${\FinVec}_!$ with a field of characteristic $p$ prime.

  For $p=2$ we recover the $3$ distinct Church numerals. But for
  $p=3$, we deduce that there are $8$ distinct Church numerals (the
  $8$ corresponding matrices are reproduced in
  Table~\ref{tab:sevenmat}). As there is almost a factorial function, the number of
  distinct Church numerals grows fast as $p$ grows: With $\mathbb{F}_5$, there
  are $64$ distinct numerals, and with $\mathbb{F}_7$ there are $426$ distinct
  numerals.
\end{example}

\begin{example}
  \label{ex:numbool-alg}
  Let us briefly reprise Example~\ref{ex:numbool} in
  the context of ${\FinVec}_!$. Even with a field of characteristic
  $2$, the vector space $\fvdenot{(\bit\to\bit)\to(\bit\to\bit)}$ is
  relatively
  large: $\bit$ has dimension 2 and consists of $4$ vectors, ${!\bit}$
  then has dimension $4$ and consists of $16$ vectors. The dimension
  of the homset ${!\bit}\loli\bit$ is $8$, and it contains $2^{8}=256$
  vectors. Using the representation of the two previous examples, a
  Church numeral is then a matrix of size $256\times 8$.
  
  Let us represent a function ${!\bit}\loli\bit$ as a tuple $({\tt
    x}^k_{ij})_{i,j,k}$ lexicographically ordered
  $
  {\tt x}^0_{00}, {\tt x}^1_{00}, {\tt x}^0_{01}, {\tt x}^1_{01}, {\tt
    x}^0_{10}, {\tt x}^1_{10}, {\tt x}^0_{11}, {\tt x}^1_{11},
  $
  representing the map sending $b_{i\cdot\ttrue+j\cdot\ffalse}$ to
  ${\tt x}^0_{ij}\cdot\ttrue + {\tt x}^1_{ij}\cdot\ffalse$.
  These form the basis elements of the range of the matrix. The
  domain of the matrix consists of all the $256$ combinations of 0/1
  values that these can take. Ordered lexicographically, they form the
  basis of the domain of the matrix.
  
  As before, the Church numeral $\bar0$ is constant while 
  $\bar{1}$ is the identity. 
  The numeral $\bar{2}$ sends each of the $8$-tuples $({\tt
    x}^k_{ij})_{i,j,k}$ to the $8$-tuple
  $
  (x^0_{x^0_{a,b},{\tt x}^1_{a,b}},{\tt x}^1_{x^0_{a,b},{\tt x}^1_{a,b}})_{a,b\in\{0,1\}},
  $
  and so forth.
  So for example, the negation sending $b_{a\cdot\ttrue +
    b\cdot\ffalse}$ to $a\cdot\ffalse + b\cdot\ttrue$ is the $8$-tuple
  $(0,0,1,0,0,1,1,1)$ and is sent by $\bar{2}$ to the tuple
  $(0,0,0,1,1,0,1,1)$ which is indeed the identity.
  
  If one performs the calculation, one finds out that in
  ${\FinVec}_!$, over the type $\bit$, there are exactly $15$ distinct
  Church numerals. The numerals $\bar0$, $\bar1$ and $\bar2$ are
  uniquely determined, and then the semantics distinguishes the
  equivalence classes $\{i+12n\bor n\in\mathbb{N}\}$, for
  $i=3,4,\ldots14$. The $14$ non-constant Church numerals are
  represented in Table~\ref{tab:14num}: First column contains numbers $1$
  to $7$, second columns numbers $8$ to $14$. The matrices are represented
  as rectangles made of $256\times8$ squares. Black squares mean $1$
  and white squares mean $0$.
\end{example}

\subsection{Properties of the \texorpdfstring{${\FinVec}$}{FinVec} Model}

\begin{table*}[b]
  \caption{The $14$ non-constant Church numerals over $\bit$ in ${\FinVec}_!$ with $K=\mathbb{F}_2$.}
  \label{tab:14num}
  \def\b{\rule{0.026cm}{0.03cm}}
  \def\a{\phantom{\rule{0.026cm}{0.03cm}}}
  \def\mynl{\\[-2.5ex]}
  \def\mynnl{\\[-1ex]}
  \scalebox{.89}{\begin{minipage}{5.32in}\[

\end{array}
\]
\end{minipage}}
\end{table*}

As shown in the next results, this semantics is both sound and
adequate with respect to the operational equivalence.
Usually adequacy uses non-terminating terms. Because the
language is strongly normalizing, we adapt the notion.
However, because there are usually more maps between $\fvdenot{A}$ and
$\fvdenot{B}$ than between $\fsdenot{A}$ and $\fsdenot{B}$ (as shown
in Examples~\ref{ex:numunit-alg}, \ref{ex:numunit-alg2}
and~\ref{ex:numbool-alg}), the model fails to be fully abstract.

\begin{lemma}
  \label{lem:axdenot-st-alg}
  If $M\axeq N:A$ then $\fvdenot{M} =
  \fvdenot{N}$.\qed
\end{lemma}

\begin{theorem}
  \label{th:sound-st-alg}
  If $\Delta\entail M,N:A$ and $\fvdenot{M}=\fvdenot{N}$ then 
  $M\opeq N$.
\end{theorem}

\begin{proof}
  The proof is similar to the proof of Theorem~\ref{th:sound} and
  proceeds by contrapositive, using Lemmas~\ref{lem:uniq-value},
  \ref{lem:sn}, \ref{lem:opax} and~\ref{lem:axdenot-st-alg}.
\end{proof}

\begin{theorem}[Adequacy]
  Given two closed terms $M$ and $N$ of type $\bit$, $\fvdenot{M}=\fvdenot{N}$ if and
  only if $M\opeq N$.
\end{theorem}

\begin{proof}
  The left-to-right direction is Theorem~\ref{th:sound-st-alg}. For
  the right-to-left direction, since the terms $M$ and $N$ are closed
  of type $\bit$, one can choose the context $C[-]$ to be $[-]$, and
  we have $M\to^* b$ if and only if $N\to^* b$. From
  Lemma~\ref{lem:sn}, there exists such a boolean $b$: we deduce from
  Lemma~\ref{lem:opax} that $M\axeq N$. We conclude with
  Lemma~\ref{lem:axdenot-st-alg}.
\end{proof}

\begin{remark}
\label{rem:notfullyabst}
The model ${\FinVec}_!$ is not fully abstract. 
Indeed, consider the two valid typing judgments
$
x:\bit\entail\ttrue:\bit
$
and
$
x:\bit\entail\ifterm{x}{\ttrue}{\ttrue}:\bit
$.
The denotations of both of these judgments are linear maps 
$!\fvdenot{\bit}\to\fvdenot{\bit}$. According to the rules of
Table~\ref{tab:denot-lc-alg}, the denotation of the first term is 
the constant function sending all non-zero
vectors $b_{-}$ to $\ttrue$.

For the second term,
suppose that $v\in{!\fvdenot{\bit}}$ is 
equal to
$\sum_i\gamma_i\cdot{}b_{\alpha_i\cdot\ttrue + \beta_i\cdot\ffalse}$.
Let $\nu = \sum_i\gamma_i(\alpha_i + \beta_i)$.
Then since $\fvdenot{x:\bit\entail x:\bit}(v) = \nu$, 
the denotation of the second term is the function sending $v$
to $\nu\cdot\fvdenot{x:\bit\entail\ttrue:\bit}(v)$, equal to
$\nu\cdot\ttrue$ from what we just discussed.
We conclude that if $v=b_0$, then $\nu=0$: the denotation of
$x:\bit\entail\ifterm{x}{\ttrue}{\ttrue}:\bit$ sends $b_0$ to $0$.

Nonetheless, they are clearly operationally equivalent in ${\bf
  PCF}_{\it f}$ since their denotation in {\FinSet} is the same.
The language is not expressive enough to distinguish between these two
functions. Note that there exists operational settings where these
would actually be different, for example if we were to allow divergence.
\end{remark}

\begin{remark}
Given a term $A$, another question one could ask is whether the set of
terms $M:A$ in ${\bf PCF}_{\it f}$ generates a free family of vectors in the
vector space $\fvdenot{A}$. It turns
out not: The field structure brought into the model introduces
interferences, and algebraic sums coming from operationally distinct
terms may collapse to a representable element. For example, supposing
for simplicity that the characteristic of the field is $q=2$, consider
the terms $T_{\ttrue,\ttrue}$, $T_{\ffalse,\ffalse}$,
$T_{\ttrue,\ffalse}$ and $T_{\ffalse,\ttrue}$ defined as $T_{{\tt
    y},{\tt z}} = \lambda x.\ifterm{x}{\tt y}{\tt z}$,
all of types $\bit\to\bit$. They are clearly operationally distinct,
and their denotations live in ${!\bit}\loli\bit$. They can be
written as a $2\times 4$ matrices along the bases
$(b_{0},b_{\ttrue},b_{\ffalse},b_{\ttrue+\ffalse})$ for the domain and
$(\ttrue,\ffalse)$ for the range. The respective images of the $4$
terms are
$\left(\begin{smallmatrix}
0&1&1&0\\
0&0&0&0
\end{smallmatrix}\right)$
,
$\left(\begin{smallmatrix}
0&0&0&0\\
0&1&1&0
\end{smallmatrix}\right)$,
$\left(\begin{smallmatrix}
0&1&0&1\\
0&0&1&1
\end{smallmatrix}\right)$,
$\left(\begin{smallmatrix}
0&0&1&1\\
0&1&0&1
\end{smallmatrix}\right)$
and clearly,
$
\fvdenot{T_{\ttrue,\ttrue}}
=
\fvdenot{T_{\ffalse,\ffalse}}
+
\fvdenot{T_{\ttrue,\ffalse}}
+
\fvdenot{T_{\ffalse,\ttrue}}.
$

So if the model we are interested in is ${\FinVec_!}$, the language is
missing some structure to correctly handle the algebraicity.
\end{remark}

\section{An algebraic lambda-calculus}
\label{sec:alglc}

\begin{table}[t]
    \caption{Rewrite system for the algebraic fragment of
      ${\bf PCF}_{\it f}^{\it alg}$.}
    \label{tab:rw-alg}
\scalebox{.83}{\begin{minipage}{5.7in}
\begin{align*}
  \alpha\cdot M + \beta\cdot M &\to (\alpha+\beta)\cdot M
  &
  M + N &\to N + M
  &
  (M + N) + P &\to M + (N + P)
  \\
  \alpha\cdot M + M &\to (\alpha + 1)\cdot M
  &
  0\cdot M \to 0&\qquad 1\cdot M \to M
  &
  \alpha\cdot(M + N) &\to \alpha\cdot M + \alpha\cdot N
  \\
  M + M &\to (1+1)\cdot M
  &
  \alpha\cdot 0 \to 0&\qquad 0 + M \to M
  &
  \alpha\cdot(\beta\cdot M) &\to (\alpha\beta)\cdot M
  &
\end{align*}\end{minipage}}
\end{table}

To solve the problem, we extend the language ${\bf PCF}_{\it f}$ by adding an
algebraic structure to mimic the notion of linear distribution
existing in ${\FinVec_!}$. The extended language ${\bf PCF}_{\it f}^{\it alg}$
is a call-by-name variation of~\cite{ad08,adv11} and reads as follows:
\[\begin{array}{rl}
  M,N,P ~ {:}{:}{=}~ & 
  x \bor \lambda x.M \bor MN \bor
  \pi_l(M) \bor \pi_r(M) \bor \pair{M,N}\bor\punit\bor
  \ttrue \bor \ffalse \bor\\& \ifterm{M}{N}{P} \bor
  \letunit{M}{N}\bor
  0 \bor M + N \bor \alpha\cdot M,
  \\
  A,B ~ {:}{:}{=}~ &
  \tunit\bor\bit \bor A \to B \bor A \times B.
\end{array}\]
The scalar $\alpha$ ranges over the field.
The values are now
$U,V~{:}{:}{=}~
x \borsmall \lambda x.M\borsmall\pair{M,N}\borsmall$
$\punit\borsmall\ttrue \borsmall \ffalse \borsmall
0 \borsmall U + V \borsmall \alpha\cdot U.
$.
The typing rules are the same for the regular constructs. The new
constructs are typed as follows: for all $A$, $\Delta\vdash 0 : A$, 
and provided that $\Delta\vdash M,N : A$,
then $\Delta\vdash M + N : A$ and $\Delta\vdash \alpha\cdot M : A.$
The rewrite rules are extended as follows.

\smallskip
\noindent
1)~~ A set of algebraic rewrite rules shown in
Table~\ref{tab:rw-alg}. We shall explicitly talk about {\em algebraic
  rewrite rules} when referring to these extended rules. The top row
consists of the associativity and commutativity (AC) rules. We shall
use the term {\em modulo AC} when referring to a rule or property that
is true when not regarding AC rules. For example, modulo AC the term $\star$ is
in normal form and $\alpha\cdot M+(N+\alpha\cdot P)$
reduces to $\alpha\cdot(M+P)+N$.
The reduction
rules from $\Gamma$ will be called {\em non-algebraic}.

\smallskip
\noindent
2)~~
 The relation between the algebraic structure and the other
  constructs: one says that a construct $c(-)$ is
  \define{distributive} when for all $M,N$, $c(M+M) \to c(M) + c(N)$, 
  $c(\alpha\cdot M) \to \alpha\cdot c(M)$ and 
  $c(0) \to 0$.
The following constructs are distributive:
$(-)P$,
$\ifterm{(-)}{P_1}{P_2}$,
$\pi_i(-)$,
$\letunit{(-)}{N}$,
and the pairing construct factors: 
$\pair{M,N} + \pair{M',N'} \to
  \pair{M + M',N + N'}$, 
$\alpha\cdot\pair{M,N} \to
  \pair{\alpha\cdot M, \alpha\cdot N}$ and
$0^{A\times B}
  \to
  \pair{0^A, 0^B}$.

\smallskip
\noindent
3)~~
 Two congruence rules. If $M\to M'$, then 
 $M+N \to M'+N$ and $\alpha\cdot M \to \alpha\cdot M'$.

\begin{remark}\label{rem:red-strat}
  Note that if $(M_1+M_2)(N_1+N_2)$ reduces to
  $M_1(N_1+N_2)+M_2(N_1+N_2)$, it does {\em not} reduce to
  $(M_1+M_2)N_1 + (M_1+M_2)N_2$. If it did, one would get an
  inconsistent calculus~\cite{ad08}. For example, the term $ (\lambda
  x.\pair{x,x})(\ttrue+\ffalse) $ would reduce both to
  $\pair{\ttrue,\ttrue}+\pair{\ffalse,\ffalse}$ and to
  $\pair{\ttrue,\ttrue}+\pair{\ffalse,\ffalse}+
  \pair{\ttrue,\ffalse}+\pair{\ffalse,\ttrue}$.
  We'll come back to this distinction in Section~\ref{sec:call-value-reduction}.
\end{remark}

The algebraic extension preserves the safety properties, the
characterization of values and the strong normalization.
Associativity and commutativity induce a subtlety.

\begin{lemma}
  \label{lem:sn-AC}
  The algebraic fragment of ${\bf PCF}_{\it f}^{\it alg}$ is strongly
  normalizing modulo AC.
\end{lemma}

\begin{proof}
  The proof can be done as in~\cite{ad08}, using the same measure on
  terms that decreases with algebraic rewrites. The measure, written
  $a$, is defined by $a(x)=1$, $a(M+N)=2+a(M)+a(N)$,
  $a(\alpha\cdot M)=1+2a(M)$, $a(0)=0$.
\end{proof}

\begin{lemma}[Safety properties mod AC]
  \label{lem:safety-prop-alg-modAC}
  A well-typed term $M:A$ is a value or, if not,
  reduces to some $N:A$ via a sequence of steps among
  which one is {\em not} algebraic.\qed
\end{lemma}

\begin{lemma}
  \label{lem:uniq-value-alg}
  Any value of type $\tunit$ has AC-normal form $0$, $\punit$ or
  $\alpha\cdot\punit$, with $\alpha\neq 0,1$.
  \qed
\end{lemma}

\begin{lemma}
  \label{lem:sn-alg}
  Modulo AC, ${\bf PCF}_{\it f}^{\it alg}$ is strongly normalizing.
\end{lemma}

\begin{proof}
  The proof is done by defining an intermediate language ${{\bf PCF}_{\it f}}_{\it
    int}$ where scalars are omitted. Modulo AC, this language is
  essentially the language $\lambda_{-{\bf w}{\rm LK}^\to}$
  of~\cite{groote94}, and is therefore SN. Any term of ${\bf PCF}_{\it f}^{\it
    alg}$ can be re-written as a term of ${{\bf PCF}_{\it f}}_{\it int}$. With
  Lemma~\ref{lem:safety-prop-alg-modAC}, by eliminating some algebraic
  steps a sequence of reductions in ${\bf PCF}_{\it f}^{\it alg}$ can be
  rewritten as a sequence of reductions in ${{\bf PCF}_{\it f}}_{\it int}$. We
  conclude with Lemma~\ref{lem:sn-AC}, saying there is always a finite
  number of these eliminated algebraic rewrites.
\end{proof}

\subsection{Operational equivalence}
\label{sec:lc-cat-alg}

As for ${\bf PCF}_{\it f}$, we define an operational equivalence on terms of the 
language
${\bf PCF}_{\it f}^{\it alg}$. A {\em context $C[-]$} for this language has the same
grammar as for ${\bf PCF}_{\it f}$, augmented with algebraic
structure: 
$C[-]~{:}{:}{=}~
  \alpha\cdot{}C[-] \bor C[-] + N \bor M + C[-] \bor 0$.

For ${\bf PCF}_{\it f}^{\it alg}$, instead of using closed contexts of type
$\bit$, we shall use contexts of type $\tunit$: thanks to
Lemma~\ref{lem:uniq-value-alg}, there are distinct normal forms for
values of type $\tunit$, making this type a good (and slightly simpler)
candidate.

We therefore say that $\Delta\entail M:A$ and $\Delta\entail N:A$ are
operationally equivalent, written $M\opeq N$, if for all closed
contexts $C[-]$ of type $\tunit$ where the hole binds $\Delta$, for
all $b$ normal forms of type $\tunit$, $C[M]\to^* b$ if and only if $
C[N]\to^* b $.

\subsection{Axiomatic equivalence}

The axiomatic equivalence on ${\bf PCF}_{\it f}^{\it alg}$ consists of the one
of ${\bf PCF}_{\it f}$, augmented with the added reduction rules.

\begin{lemma}
  \label{lem:opax-alg}
  If $M:A$ and $M\to N$ then $M\axeq N$.\qed
\end{lemma}

\subsection{Finite vector spaces as a model}
\label{sec:model-alglc}

The category $\FinVec_!$ is a denotational model of the language
${\bf PCF}_{\it f}^{\it alg}$. Types are interpreted as for the language
${\bf PCF}_{\it f}$ in Section~\ref{sec:lc-finvec-model}. Typing judgments are
also interpreted in the same way, with the following additional
rules. First,
$\fvdenot{\Delta\entail 0:A} = 0$. Then $\fvdenot{\Delta\entail \alpha\cdot M:A} = \alpha\cdot
  \fvdenot{\Delta\entail M:A}$. Finally, we have
$\fvdenot{\Delta\entail M+N:A} = \fvdenot{\Delta\entail M:A} + \fvdenot{\Delta\entail N:A}$.

\begin{remark}
With the extended term constructs, the language ${\bf PCF}_{\it f}^{\it alg}$
does not share the drawbacks of ${\bf PCF}_{\it f}$ emphasized in
Remark~\ref{rem:notfullyabst}. In particular, the two valid typing
judgments $ x:\bit\entail\ttrue:\bit $ and $
x:\bit\entail\ifterm{x}{\ttrue}{\ttrue}:\bit $ are now operationally
distinct. For example, if one chooses the context $C[-]=(\lambda x.[-])0$, the
term $C[\ttrue]$ reduces to $\ttrue$ whereas the term
$C[\ifterm{x}{\ttrue}{\ttrue}]$ reduces to $0$.
\end{remark}

\begin{lemma}
  \label{lem:axdenot-alg}
  If $M\axeq{}N:A$ in ${\bf PCF}_{\it f}^{\it alg}$ then $\fvdenot{M} =
  \fvdenot{N}$.\qed
\end{lemma}

\begin{theorem}
  \label{th:sound-alg}
  Let $\Delta\entail M,N:A$ be two valid typing judgments in
  ${\bf PCF}_{\it f}^{\it alg}$. If $\fvdenot{M}=\fvdenot{N}$ then we also have
  $M\opeq N$.
\end{theorem}

\begin{proof}
  The proof is similar to the proof of Theorem~\ref{th:sound}: Assume
  $M\not\opeq N$. Then there exists a context $C[-]$ that distinguishes
  them. The call-by-name reduction preserves the type from
  Lemma~\ref{lem:safety-prop-alg-modAC}, and $C[M]$ and $C[N]$ can be rewritten
  as the terms
  $(\lambda y.C[y\,x_1\ldots x_n])\lambda x_1\ldots x_n.M$ and $(\lambda
  y.C[y\,x_1\ldots x_n])\lambda x_1\ldots x_n.N$, and these are axiomatically
  equivalent to distinct normal forms, from Lemmas~\ref{lem:sn-alg}
  and~\ref{lem:opax-alg}. We conclude from Lemmas~\ref{lem:opax-alg}
  and~\ref{lem:axdenot-alg} that the
  denotations of $M$ and $N$ are distinct.
\end{proof}

\subsection{Two auxiliary constructs}
\label{sec:aux-constructs}

Full completeness requires some machinery.  It is obtained by
showing that for every type $A$, for every vector $v$ in
$\fvdenot{A}$, there are two terms $M^A_v : A$ and $\delta^A_v : A \to
\tunit$ such that $\fvdenot{M^A_v} = v$ and $\fvdenot{\delta^A_v}$ sends
$b_{v}$ to $\punit$ and all other $b_{-}$'s to $0$.

We first define a family of terms ${\tt exp}^i:\tunit\to\tunit$
inductively on $i$: ${\tt exp}^0 = \lambda x.\star$ and 
${\tt exp}^{i+1} = \lambda x.\letunit{x}{{\tt exp}^i(x)}$.
One can show that $\fvdenot{{\tt
  exp}^{i}(\alpha\cdot\star)} = \alpha^i\cdot\star$.
Then assume that $o$ is the
order of the field. Let  ${\tt iszero}:\tunit\to\tunit$ be the term 
${\tt exp}^o$. The denotation of ${\tt
  iszero}$ is such that $\fvdenot{{\tt izero}(\alpha\cdot\star)}=0$ if
$\alpha=0$ and $\star$ otherwise.

The mutually recursive definitions of $\delta^A_v$ and $M^A_v$ read as
follows.

\smallskip
\noindent
{\em At type $A=\tunit$.}~~
The term $M^{\tunit}_{\alpha\cdot\star}$ is simply $\alpha\cdot\star$.
The term $\delta^{\tunit}_{\alpha\cdot\star}$ is 
$
\lambda x.{\tt iszero}(x-\alpha\cdot\star)
$.

\smallskip
\noindent
{\em At type $A=\bit$.}~~
As for the type $\tunit$, the term $M^{\bit}_{\alpha\cdot\ttrue+\beta\cdot\ffalse}$
is simply $\alpha\cdot\ttrue+\beta\cdot\ffalse$. The term
$\delta^{\bit}_{\alpha\cdot\ttrue+\beta\cdot\ffalse}$ is reusing the
definition of $\delta^\tunit$: it is the term
$
\lambda
x.\letunit{\delta^{\tunit}_{\alpha\cdot\star}(\ifterm{x}{\star}{0})}{}$
$
\delta^{\tunit}_{\beta\cdot\star}(\ifterm{x}{0}{\star})
$.

\smallskip
\noindent
{\em At type $A=B\times C$}.~~
If
$v\in\fvdenot{A}={\fvdenot{B}}\times{\fvdenot{C}}$, then $v =
\pair{u,w}$, with $u\in\fvdenot{B}$ and $w\in\fvdenot{C}$. 
By induction, one can construct $M^B_{u}$ and $M^C_{w}$: the term
$M^{B\times C}_v$ is $\pair{M^B_{u},M^C_{w}}$. Similarly, one can construct the terms
$\delta^B_u$ and $\delta^C_w$: the term $\delta^{B\to C}_v$ is 
$
\lambda x.\letunit{\delta^B_u\,\pi_l(x)}{\delta^C_w\,\pi_r(x)}$.

\smallskip
\noindent
{\em At type $A=B\to C$.}~~
Consider
$f\in\fvdenot{A}={!\fvdenot{B}}\loli\fvdenot{C}$. The domain of $f$
is finite-dimensional: let $\{b_{u_i}\}_{i=1\ldots n}$ be its basis, and let $w_i$
be the value $f(b_{u_i})$. Then, using the terms $\delta^B_{u_i}$ and
$M^C_{w_i}$, one can define $M^{B\to C}_v$ as the term
$
\sum_i\lambda x.\letunit{\delta^{B}_{u_i}\,x}{M^{C}_{w_i}}
$.
Similarly, one can construct $\delta^C_{w_i}$ and $M^B_{u_i}$, and from
the construction in the previous paragraph we can also generate
$\delta^{C^{\times n}}_{\pair{w_1,\ldots w_n}}:C^{\times n}\to\bit$. The term
$\delta^{B\to C}_v$ is then defined as
$
\lambda f.\delta^{C^{\times n}}_{\pair{w_1,\ldots w_n}}\,\pair{f\,M^B_{u_1},\ldots,f\,M^B_{u_1}}.
$

\subsection{Full completeness}

We are now ready to state completeness, whose proof is simply by
observing that any $v\in\fvdenot{A}$ can be realized by the term $M^A_v:A$.

\begin{theorem}[Full completeness]
  \label{th:comp-alg}
  For any type $A$, any vector $v$ of $\fvdenot{A}$ in $\FinVec_!$ is
  representable in the language ${\bf PCF}_{\it f}^{\it alg}$.\qed
\end{theorem}

\begin{theorem}
  \label{th:eq-alg}
  For all $M$ and $N$, $M\opeq N$ if and only if
  $\fvdenot{M}=\fvdenot{N}$.\qed
\end{theorem}

A corollary of the full completeness is that the semantics $\FinVec$ is also adequate
and fully abstract with respect to ${\bf PCF}_{\it f}^{\it alg}$.

\section{Discussion}
\label{sec:discussion}

\subsection{Simulating the vectorial structure.}
\label{sec:fact-alg}

As we already saw, there is a full embedding of category $E:
{\FinVec}_{!} \hookrightarrow {\FinSet}$. This embedding can be
understood as ``mostly'' saying that the vectorial structure ``does not
count'' in ${\FinVec}_!$, as one can simulate it with finite sets.
Because of Theorems~\ref{th:eq} and~\ref{th:eq-alg}, on the syntactic
side algebraic terms can also be simulated by the regular ${\bf
  PCF}_{\it f}$.

\begin{table}[b]
  \caption{Relation between ${\bf PCF}_{\it f}$ and ${\bf PCF}_{\it f}^{\it alg}$}
  \label{tab:l-l-alg}
  \scalebox{.93}{\begin{minipage}{5.1in}
  \[
  \begin{array}{@{}c@{}}
  \phi^{\rm vec}_{\tunit} = \letunit{\delta_{0}x}{\ttrue} +
  \letunit{\delta_{\star}x}{\ffalse}
  \qquad
  \bar\phi^{\rm vec}_{\tunit} = \ifterm{x}{0}{\star}
  \\
  \begin{array}{l@{}l@{}l}
    \phi^{\rm vec}_{\bit} {}={}& 
    \phantom{{}+{}}\letunit{\delta_{0}x}{\pair{\ttrue,\ttrue}}
    &{}+\letunit{\delta_{\ttrue}x}{\pair{\ttrue,\ffalse}}
    \\
    &
    {}+\letunit{\delta_{\ffalse}x}{\pair{\ffalse,\ttrue}}
    &{}+\letunit{\delta_{\ttrue+\ffalse}x}{\pair{\ffalse,\ffalse}}
  \end{array}
  \\
  \begin{array}{l@{}l}
    \bar\phi^{\rm vec}_{\bit} {}={}&
    {\tt if}~(\pi_l{}x)\,{\tt then}\,(\ifterm{(\pi_r{}x)}{0}{\ttrue})
    {\tt else}~(\ifterm{(\pi_r{}x)}{\ffalse}{\ttrue+\ffalse})
  \end{array}
  \\
  \begin{array}{l@{}l}
    \phi^{\rm vec}_{B\times C} {}={}
    &
    \pair{x;\pi_l;\phi^{\rm vec}_{B}, x;\pi_r;\phi^{\rm vec}_{C}},
    \\
    \bar\phi^{\rm vec}_{B\times C} {}={}
    &
    \pair{x;\pi_l;\bar\phi^{\rm vec}_{B}, x;\pi_r;\bar\phi^{\rm vec}_{C}},
  \end{array}
  \qquad
  \begin{array}{@{}l@{}l}
    \phi^{\rm vec}_{B\to C} {}={}
    &
    \lambda y.x(y;\bar\phi^{\rm vec}_{B});\phi^{\rm vec}_{C},
    \\
    \bar\phi^{\rm vec}_{B\to C} {}={}
    &
    \lambda y.x(y;\phi^{\rm vec}_{B});\bar\phi^{\rm vec}_{C}.
  \end{array}
  \end{array}
  \]
  \end{minipage}}
\end{table}

In this section, for simplicity, we assume that the field
is $\mathbb{F}_2$. In general, it can be any finite size provided
that the regular lambda-calculus ${\bf PCF}_{\it f}$ is augmented with $q$-bits,
i.e. base types with $q$ elements (where $q$ is the characteristic of
the field).

\begin{definition}
  The \define{vec-to-set} encoding of a type $A$, written
  $\vectoset{A}$, is defined inductively as follows:
  $\vectoset(\tunit)=\bit$,
  $\vectoset(\bit)=\bit\times\bit$, $\vectoset(A\times
  B)=\vectoset(A)\times\vectoset(B)$, and $\vectoset(A\to
  B)=\vectoset(A)\to\vectoset(B)$.
\end{definition}

\begin{theorem}
  \label{th:l-l-alg}
  There are two typing judgments $x:A \vdash \phi^{\rm vec}_A :
  \vectoset(A)$ and $x:\vectoset(A)\vdash
  {\bar\phi}^{\rm vec}_A : A$, inverse of each other, in ${\bf PCF}_{\it f}^{\it
    alg}$ such that any typing judgment $x:A\vdash M:B$ can be
  factored into
  $
  A 
  \xrightarrow{\phi^{\rm vec}_A}
  \vectoset(A)
  \xrightarrow{\tilde{M}}
  \vectoset(B)
  \xrightarrow{\bar\phi^{\rm vec}_B}
  B,
  $
  where $\tilde{M}$ is a regular lambda-term of ${\bf PCF}_{\it f}$.
\end{theorem}
\begin{proof}
The two terms $\phi^{\rm vec}_A$ and $\bar\phi^{\rm vec}_A$ are
defined inductively on $A$.
For the definition of $\phi^{\rm vec}_{\bit}$ we are
reusing the term $\delta_v$ of Section~\ref{sec:aux-constructs}.
The definition is in Table~\ref{tab:l-l-alg}
\end{proof}

\subsection{Categorical structures of the syntactic categories.}
\label{sec:categ-struct-synt}

Out of the language ${\bf PCF}_{\it f}$ one can define a syntactic 
category: objects are
types and morphisms $A\to B$ are valid typing judgments $x:A\vdash M:B$ modulo
operational equivalence. 
Because of Theorem~\ref{th:eq}, this category is
cartesian closed, and one can easily see that the product of $x:A\entail M:B$
and $x:A\entail N:C$ is $\pair{M,N}:B\times C$, that the terminal
object is $\punit:\tunit$, that
projections are defined with $\pi_l$ and $\pi_r$, and that the
lambda-abstraction plays the role of the internal morphism.

The language ${\bf PCF}_{\it f}^{\it alg}$ almost defines a cartesian closed category: by
Theorem~\ref{th:eq-alg}, it is clear that pairing and lambda-abstraction form a
product and an internal hom. However, it is missing a terminal object
(the type $\tunit$ doesn't make one as $x:A\vdash 0:\tunit$ and
$x:A\vdash \star:\tunit$ are operationally distinct). There is no type
corresponding to the vector space $\langle0\rangle$. It is not
difficult, though, to extend
the language to support it: it is enough to only add a type $\bottype$. Its only
inhabitant will then be the term $0$: it make a terminal
object for the syntactic
category.

Finally, Theorem~\ref{th:l-l-alg} is essentially giving us a functor
${\bf PCF}_{\it f}^{\it alg}\to{{\bf PCF}_{\it f}}$ corresponding to the full embedding
$E$. This makes
a full correspondence between the two models $\FinSet$ and $\FinVec_!$, and
${\bf PCF}_{\it f}$ and ${\bf PCF}_{\it f}^{\it alg}$, showing that computationally the algebraic
structure is virtually irrelevant.

\subsection{(Co)Eilenberg-Moore category and call-by-value}
\label{sec:call-value-reduction}

From a linear category with modality $!$ there are two canonical cartesian
closed categories: the coKleisli category, but also the (co)Eilenberg-Moore
category: here, objects are still those of ${\FinVec}$, but morphisms
are now ${!A}\to{!B}$.

According to~\cite{valiron13typed}, such a model would correspond to
the call-by-value (or, as coined by~\cite{diazcaro-phd} {\em call-by-base})
strategy for the algebraic structure discussed in
Remark~\ref{rem:red-strat}.

\subsection{Generalizing to modules}
\label{sec:generalizing-modules}

To conclude this discussion, let us consider a generalization of finite vector
spaces to finite modules over finite semi-rings.

Indeed, the model of linear logic this paper uses would work in the context of a
finite semi-ring instead of a finite field, as long as addition and
multiplication have distinct units. For example, by using the
semiring $\{0,1\}$ where $1+1=1$ one recover sets and relations.
However, we heavily rely on the fact
that we have a finite field $K$ in the construction of
Section~\ref{sec:aux-constructs}, yielding the completeness result in
Theorem~\ref{th:comp-alg}.

This particular construction works because one can construct any
function between any two finite vector spaces as polynomial, for the
same reason as any function $K\to K$ can be realized as a polynomial.

\end{document}
